\documentclass{amsart}
\usepackage{amssymb}
\usepackage{graphicx}
\usepackage{amscd}

\newtheorem{theorem}{Theorem}
\theoremstyle{plain}

\numberwithin{equation}{section}

\begin{document}
\title[Weighted mean losses or gains]{Functional weak laws for the weighted
mean losses or gains and applications}
\author{Gane Samb Lo$^{*}$, Serigne Touba Sall$^{**}$, Pape Djiby Mergane $^{***}$}
\address{$^{*}$ Lerstad and Universit\'{e} Gaston Berger. BP 234,
Saint-Louis, S\'{e}n\'{e}gal \\
and LSTA, UMPC, Paris VI\\
, $^{**}$ LERSTAD and Ecole Normale Sup\'{e}rieure, Dakar, S\'{e}n\'{e}gal. }
\email{gslo@ufrsat.org, stsall@ufrsat.org\\
$^{***}$ Pape Djiby Mergane, Lerstad and Universit\'{e} Gaston Berger. BP 234,
Saint-Louis, S\'{e}n\'{e}gal.}
\subjclass[2000]{Primary 60F05, 60F15; Secondary 91C05}
\keywords{Empirical process, time dependent process, weak theory, risk measures, poverty index, loss function, economic welfare}

\begin{abstract}
In this paper, we show  that many risk measures arising in Actuarial Sciences, Finance, Medicine, Welfare analysis, etc. are garthered in
classes of Weighted Mean Loss or Gain (WMLG) statistics. Some of them are Upper Threshold Based (UTH) or Lower Threshold Based (LTH). These
statistics may be time-dependent when the scene is monitored in the time and depend on specific functions $w$ and $d$. This paper provides 
time-dependent and uniformly functional weak asymptotic laws that allow temporal and spatial studies of the risk as well as comparison between statistics in terms of dependence and mutual influence. The results are particularised for usual statistics like the Kakwani and Shorrocks ones that are mainly used in welfare analysis. Datadriven applications based on pseudo-panel data are provided.

\end{abstract}

\maketitle

\Large

\section{Introduction and motivation}

In many situations and many areas, we face the double problem of estimating
the risk of lying in some marked zone and, at the same time, the cost
associated with it. To fix ideas, we may be interessed in estimating the
immunocompromised patients number $Q$, the size of the set $\mathcal{M}$\ of
infected people, in some population $\mathcal{P}$. At the same time, we
know that the severity of the infection is measured by the viral load $Y$\
expressed in RNA copies per milliliter of blood plasma. The cost of
treatement, for example a course of chemotherapy, heavily depends on the viral
load. If one has to treat all the patients, there is a cost to pay for each
treatment, that is a cost function $d(Y).$ Facing these two problems at the
same time, comparing two different populations or monitoring the evolution
of the global situation should be based on the couple $(Q,d(Y))$ rather than
on which is commonly called the HIV/AIDS adult prevalence rate, on what is
based international comparison. In order to make a workable statistic,
consider a sample of individuals $\mathcal{E}=\{1,2,...,n\}$\ drawn for $%
\mathcal{P}$ and measure the viral load $Y_{j}$ for each $j\in \mathcal{M}.$ A
general comparative statistic should be of the form%
\begin{equation*}
\sum_{j\in \mathcal{M}}d(Y_{j}).
\end{equation*}%
Since comparisons over the time are based on this index, one would be
interested in putting more or less emphasis on the more infected or not, in
terms of viral load. This is achieved by affecting a weight $\rho (j)$ to $%
j\in \mathcal{M}$ as a monotone function of the rank $R_{j,n}$ of $Y_{j}$ in
the sample. For an increasing $\rho $, it is paid more attention to less
infected while the contrary holds for a decreasing one. This leads to
statistics like%
\begin{equation}
J_{n}(\mathcal{M},d,\rho )=\frac{1}{n}\sum_{j\in \mathcal{M}}\rho
(R_{j,n})d(Y_{j}).  \label{ind00}
\end{equation}%
It is also known that the viral load is detectable only above a threshold of value $%
Z_{0}=40$ RNA copies per milliliter of blood plasma. We thus have%
\begin{equation*}
j\in \mathcal{M}\iff Y_{j}\geq Z_{0}
\end{equation*}%
and%
\begin{equation*}
J_{n}(\mathcal{M},d,\rho )=\frac{1}{n}\sum_{Y_{j}\geq Z_{0}}\rho
(R_{j,n})d(Y_{j}).
\end{equation*}%

\noindent We may decide to concentrate on the very expansive chemotherapy courses due to financial pressure. In that case, we change the threshold to $Z>Z_{0}$ accordingly to the available budget.\\

\noindent Such statistics are also used in insurance theory. Suppose that one
insurance company receives $n$ claims $\{Y_{1},...,Y_{n}\}.$ We may fix a
threshold $Z$ such that any claim greater than $Z$ is seen as causing a loss 
$(Y_j-Z)$\ for the company. It then becomes interesting to estimate\ the
number of possible claims over $Z$, 
\begin{equation}
Q_{n}=\sum_{j=1}^{n}1_{(Y_{j}\geq Z)}  \label{qsup}
\end{equation}%
and to choose a distorsion function $\gamma $\ of the individual loss $%
(Z-Y_{j})$, hence (\ref{ind00}) is transformed here into%
\begin{equation*}
J_{n}(Z,d,\rho )=\frac{1}{n}\sum_{Y_{j}\geq Z}\rho (R_{j,n})\gamma (Y_{j})=%
\frac{1}{n}\sum_{j\geq n-Q}\rho (j/n)d(Y_{j,n}-Z),
\end{equation*}%
where $Y_{1,n}\leq ...\leq Y_{n,n}$ are the order statistics based on $%
\{Y_{1},...,Y_{n}\}.$ In this case, $J_{n}(\mathcal{M},d,\rho )$ may be seen
as a risk measure.\\

\bigskip

\noindent In poor countries, an individual is considred as a poor one when his income $Y$ below some threshold 
$Z$, called poverty line. And then 
\begin{equation}
Q_{n}=\sum_{j=1}^{n}1_{(Y_{j}\leq Z)}  \label{qinf}
\end{equation}%
is the total number of poor people in the sample, while $Q_{n}§n$ is the poor headcount. Usually the cost function here depends on the
relative poverty gap $\gamma (Y_{j})=(Z-Y_{j})/Z.$ In this field, following Lo \cite{loGPI2013}, $\ J_{n}(M,d,\rho )$ may be called a General Poverty Index
(GPI). The same form may also be used in medical science when
dealing with vitamine (say vitamine D) deficiency. In this case, $%
J_{n}(Z,d,\rho )$ is used as a general measure of vitamine deficiency to
evaluate the mean cost of vitamine supply as a treatment.\\

\noindent We see from the lines above that (\ref{ind00}) is a very general statistic,
that works in various fields, with losses or gains dependent on the meaning
of the cost function $c$. We are entitled to name it as a Weighted Mean Loss
or Gain ($WMLG$) statistic or random measure or index. It might take a
specific name, depending on the particular field where it operates. In the loss (resp. gain) case, we simply denote it $WML$ (resp. $WMG$).\\

\noindent When we have time-dependent data, over the time $[0,T]$ with continuous
observations $\left( \{Y_{1}(t),...,Y_{n}(t)\},t\in T\right) ,$ we are led
to a time-dependent $WMLG$ statistic in the form%
\begin{equation*}
J_{n}(\mathcal{M},\gamma ,\rho ,t)=\frac{1}{n}\sum_{j\in \mathcal{M}}\rho
(R_{j,n}(t))d(Y_{j}(t)).
\end{equation*}%
In the case where $\mathcal{M}$ is based on the threshold $Z$; the latter
should eventually depend on the time and becomes $Z=Z(t).$ Also, in an
spatial analysis, it would be possible to have a particular threshold for
any area.\\

\bigskip

\noindent The choice of $d$ and $\rho $ depends of the specific role played by (%
\ref{ind00}). But, a set of axioms, which are desirable or
mandatory to be fulfilled for a welfare or a risk measure, is usually adopted. For risk
measures, such axiomes alongside an axiomatic foundation are to be found in \textit{Artzner et al.} \cite{artzner}.  For poverty analysis, a large and deep review of the axiomatic approach, due to Sen \cite{sen}, is available in Zheng \cite{zheng}.\\

\bigskip

\noindent  Finally, on taking into account various forms of (\ref{ind00}) in the literature, the following form of threshold-based weighted mean loss seems to
be a general one%
\begin{equation}
J_{n}(Z,\omega ,d)=\frac{A(n)}{nB_{n}(Q_{n})}\sum_{j=1}^{Q_{n}}w(\mu
_{1}n+\mu _{2}Q_{n}-\mu _{3}j+\mu _{4})\;d\left((Z-Y_{j,n})/Z\right)  \label{wmlgs},
\end{equation}%
or the following%
\begin{equation}
J_{n}^{\ast }(Z,\omega ,d)=\frac{A(n)}{nB_{n}(Q_{n})}\sum_{j\geq
n-Q}^{n}w(\mu _{1}n+\mu _{2}Q_{n}-\mu _{3}j+\mu _{4})\;d\left((Z-Y_{j,n})/Z\right),
\label{wmlgi}
\end{equation}

\noindent depending on whether we handle loss (with $Q_{n}$\ defined in (\ref{qinf}))
or gains (with $Q_{n}$\ defined in (\ref{qsup})), and where%
\begin{equation*}
B(Q_{n})=\sum_{j=1}^{Q_{n}}w(j).
\end{equation*}%
From a mathematical \ point of view, the asymptotic behaviors of the two
forms radically differ although the writing seems symetrical. The reason
is that for the first, the random variables used in (\ref{wmlgs}) are bounded
and the asymptotic handling is much easier. As for (\ref{wmlgi}), we should
face heavy tail problems and further complications may arise.\\

\bigskip

\noindent This paper is aimed at offering a full functional weak theory according to
the most recent setting of such theories as stated in (\cite{vaart}).
Particularly, we are interested here in the time-dependent investigation of (%
\ref{wmlgs}), and next the functional weak theory in $d$ and $w$. We call
the first class of statistics Upper Threshold Based Weighted Mean Loss or
Gain (UTB WMLG) ones and the others are named Lower Threshold Based
Weighted Mean Loss indices (LTB WMLG). This paper is only concerned with the first
class of statistics. The others will be objects of further studies.

\bigskip

\noindent Consider for a while that $w$ and $d$ are fixed as well as the time. We
notice that asymptotic results of \ $J_{n}(Z,\omega ,d)$ are available for
specific forms in Welfare theory or in Actuarial Sciences. For example, Lo(\cite{loGPI2013})
proved that%
\begin{equation*}
J_{n}(Z,\omega ,d)\rightarrow J(Z,\omega
,d)=\int_{0}^{Z}w_{G}(Z,y)d((z-y)/z)dy
\end{equation*}

\begin{equation*}
=\mathbb{E}_{G}(w_{G}(Z,Y)d((Z-Y)/z)\mathbb{I}(Y\leq Z))=\mathbb{E}wmgl,
\end{equation*}%
where $\mathbb{E}wmgl$ may be called the Exact UTB WMLG.
For instance, the weight $w_{G}(Z,y)=2(1-y)$ is related to the Shorrocks 
\cite{shorrocks} and Thon \cite{thon} statistics , $w_{G}(Z,y)=2(1-y/Z)^{k}$
is the Kakwani weight (see \cite{kakwani}), that includes the Sen \cite{sen}
one corresponding to $k=1.$ For $w_{G}(Z,y)=1,$ we get the nonweighted mean
losses or gains.\\ 

\bigskip 

\noindent To be able to base statistical tests of such results, we may be interested in finding the asymptotic law of $$\sqrt{n}%
(J_{n}(\omega ,d)-J(\omega ,d))\, \textrm{ as } n\rightarrow \infty.$$

%\bigskip

\noindent However, we still need to handle longitidunal data, where the risk situation
is analysed over a continuous period of time $[0,T]$. In this case, we are
faced with continuous data in the form of $\{Y(t),0<t<T\}$, and some
modification is needed in the definition of indices to take this into
account. We are then led to consider the time-dependent and UTB WMLG
statistic defined by

\begin{equation}
J_{n}(t)=\frac{A(Q_{n}\left( t\right) ,n,Z(t))}{nB(Q_{n}(t))}\sum_{j=1}^{Q_{n}%
\left( t\right) }w(\mu _{1}n+\mu _{2}Q_{n}\left( t\right) -\mu _{3}j+\mu
_{4})d\left( \frac{Z(t)-Y_{j,n}\left( t\right) }{Z(t)}\right) ,  \label{ssl01a}
\end{equation}%
with $0\leq t\leq T$\ and $T\in \mathbb{R}.$

\bigskip

\noindent Instead of analysing such UTB WMLG for some specific functions $w$ or $d,$
or at a fixed point $t$, it may be more valuable to have at once a uniform
weak theory on $w$, $d$ and $t\in T$. Such a result will provide individual
tests, and enables spatial and temporal comparisons of the risk measure. As well,
since all the measure are expressed in the same Gaussian field, we have
joint asymptotic distributions of the different indices themselves.\\

\noindent This paper is aimed at settling the uniform weak convergence of such
statistics, that is the asymptotic theory of the time-dependent poverty
measures (\ref{ssl01a}), in the space $C([0,T])$ of real continuous
functions defined on $[0,T]$. First attempts were treated for the special case
of time-dependent nonweighted mean loss or gain (MLG) measures in \cite{sall-lo}
and, in \cite{lodurban}, for nonrandomly $WLMG$ statistics, that is,  $WLMG$ statistics for which the
weight is nonrandom, like the Shorrocks one, is dealt with. Now, we target to give here
the most general results on the time-dependent UTB-WMLG statistics. Two potential applications areas here
are vitamine deficiency risk measures and poverty measures. It is then natural
to consider a threshold depending on the time. But we suppose that it lies
in some finite interval 
\begin{equation*}
0<Z_{1}\leq Z(t)\leq Z_{2}<+\infty .
\end{equation*}%
An important application is the statistical estimation of the Relative Mean
Loss Variation (RMLV) from time $t$ to $s$ defined as follow 
\begin{equation*}
\Delta RJ(t,s)=(J(s)-J(t))/J(t)
\end{equation*}%
by confidence intervals where $J_{n}(t)$ is a poverty
measure, one of the Millennium Development Goals (MDG) is halving of extreme
poverty from $t=2000$ to time $s=2015.$ This means that we target to have $%
\Delta RJ(t,s)\leq -50\%.$ Our results below tackle this issue.\\

\bigskip

\noindent We will need a number of hypotheses towards an adequate frame for our study.
These hypotheses may appear severe and numerous, at first sight, but most of
them are natural and easy to get. We first need the following shape conditions for the
$WMLG$ measures themselves. The letter $S$ in the hypotheses names refers to shape conditions.

\begin{itemize}
\item[(HS1)] There exist functions $h(p,q)$ of $(p,q)\in \mathbb{N}^{2},$ $%
c(u,v)$ and $\pi (u,v)$ of $(u,v)\in (0,1)^{2}$ independent of $t\in \lbrack
0,T],$ such that, as $n\rightarrow +\infty $, 
\begin{equation*}
\sup_{t\in \lbrack 0,T]}\max_{1\leq j\leq Q_{n}(t)}\left|
A(n,Q_{n}(t))h^{-1}(n,Q_{n}(t))w(\mu _{1}n+\mu _{2}Q_{n}(t)-\mu _{3}j+\mu
_{4})\right.
\end{equation*}
\begin{equation*}
\left. -c(Q_{n}(t)/n,j/n)\right| =o_{P}^{\ast }(n^{-1/2}).
\end{equation*}

\item[(HS2)]  
\begin{equation*}
\sup_{t\in \lbrack 0,T]}\max_{1\leq j\leq Q_{n}(t)}\left|
w(j)h^{-1}(n,Q_{n}(t))-\frac{1}{n}\pi (Q_{n}(t)/n,j/n)\right|
=o_{P}^{\ast }(n^{-3/2})
\end{equation*}

\item[(HS3)] There exists a function $c(u,v)$ of $(u,v)\in (0,1)^{2}$
independent of $t\in \lbrack 0,T],$ such that, as $n\rightarrow +\infty $, 
\begin{equation*}
\sup_{t\in \lbrack 0,T]}\max_{1\leq j\leq Q_{n}(t)}\left|
A(n,Q_{n}(t))B^{-1}(n,Q_{n}(t))w(\mu _{1}n+\mu _{2}Q_{n}(t)-\mu _{3}j+\mu
_{4})\right.
\end{equation*}
\begin{equation*}
\left. -c(Q_{n}(t)/n,j/n)\right| =o_{P}^{\ast }(n^{-1/2}).
\end{equation*}
\end{itemize}

\bigskip

\noindent We will require other assumptions depending on the regularity of the functions $c$ and $%
\pi$. The letter $R$ in these hypotheses name refers to Regularity conditions..

\begin{itemize}
\item[(HR1)] The bivariate functions $c$ and $\pi $ have equi-continuous
partial differential on $(\beta ,\xi )\times (0,1)$, where $\beta $ and $\xi 
$ are two real numbers to be defined later on.

\item[(HR2)] For a fixed $x$, the functions $y\rightarrow \frac{\partial c}{%
\partial y}(x,y)$ and $y\rightarrow \frac{\partial \pi }{\partial y}(x,y)$
are monotone.

\item[(HR3)] There exist $H_{0}>0$ and $H_{\infty }<+\infty $ such that, for 
$t\in \lbrack 0,T],$\newline
\begin{equation*}
H_{0}<H_{c}(t)=\int c(G_{t}(Z(t)),G_{t}(y))\gamma _{t}(y)dG_{t}(y)<H_{\infty
}
\end{equation*}
and 
\begin{equation*}
H_{0}<H_{\pi }(t)=\int \pi (G_{t}(Z(t)),G_{t}(y))e_{t}(y)dG_{t}(y)<H_{\infty
}.
\end{equation*}
\end{itemize}

\bigskip

\noindent Our final achievement is that, when putting $J(t)=H_{c}(t)/H_{\pi }(t),$ we
are able to get the uniform asymptotic law of $\{\sqrt{n}(J_{n}(t)-J(t)),0%
\leq t\leq T\}$ and to describe the limiting Gaussian process $\{\mathbb{G}%
(t),0\leq t\leq T\}.$ This enables the statistical uniform estimates of \ $%
\Delta _{n}J_{n}(t,s)=J_{n}(t)-J_{n}(s)$ by $\Delta J(t,s)=J(t)-J(s)$ by
interval confidences. We also particularize the results for the so-important
Kakwani class of WMLG statistics of which the Sen one is a member. The results
that have directly been derived for the Shorrocks case are rediscovered here.

\section{Our results}

\label{sec1}

Our results will rely on the representation of Theorem \cite{lotools2},
which in turn will need the following assumptions.

\bigskip

\begin{itemize}
\item[(HL1)] There exist $\beta >0$ and $0<\xi <1$ such that 
\begin{equation*}
0<\beta <\inf_{0\leq t\leq T}G_{t}(Z_{1})<\sup_{0\leq t\leq
T}G_{t}(Z_{2})<\xi <1.
\end{equation*}

\bigskip

\item[(HL2)] The subclass $\mathcal{F}_{0}=\{\pi _{t,Z}:x\leadsto
1_{(x(t)\leq Z)},t\in \lbrack 0,T]\}$ of $\ell ^{\infty }(C([0,T]))$, the
set of real bounded and continuous functions, is a $\mathbb{P}_{Y}-$%
Glivenco-Cantelli class, that is, as $n\rightarrow \infty ,$%
\begin{equation*}
\sup_{t\in \lbrack 0,T]}\left| G_{t,n}(Z(t))-G_{t}(Z(t))\right| \rightarrow
0,\text{ }a.s.o.p.
\end{equation*}
where, for any $t\in \lbrack 0,T]$ and $y\in \mathbb{R}$, $%
G_{t,n}(y)=n^{-1}\sum_{i=1}^{n}1_{(Y(t)\leq y)}$. As a reminder $\mathbb{R} \ni z_{n} \rightarrow 0$, $a.s.o.p$ \text{ as } $n\rightarrow +\infty$ means 
($z_{n} \rightarrow 0$ in outer probability), that is : there exists a sequence of measurable random variables, $u_n$ such that for any $n\geq 1$, $ |z_n| \leq u_n$ and $u_n \rightarrow 0 \text{ as } n\rightarrow +\infty$.\\ 

\noindent Finally let us denote $f_{t}(x)=x(t),$ where $x\in \ell ^{\infty }(C([0,T])).$

\bigskip

\item[(HL3)] For any $t\,\in\,[0,T],$ $G_{t}$ is strictly increasing and the
functions $G_{t}$ are uniformly continuous in $t\in \lbrack 0,T \rbrack $.

\bigskip

\item[(HL4)] $d$ is bounded by one and is differentiable with derivative
function $d^{\prime }$\ bounded by M : $0\leq d\leq 1,\!\left| d^{\prime
}\right| \leq M.$
\end{itemize}

\bigskip

\begin{theorem}
\label{theo1} Suppose that (HS1)-(HS2), (HR1)-(HR3) and (HL1)-(HL4) hold.
Put $J(t)=H_{c}(t)/H_{\pi }(t),$%
\begin{equation}
K_{c}(t)=\int_{0}^{1}\frac{\partial c}{\partial x}(G_{t}(Z(t)),s)\gamma_t
(G_{t}^{-1}(s))ds,\text{ }K_{\pi }(t)=\int_{0}^{1}\frac{\partial \pi }{%
\partial x}(G_{t}(Z(t)),s)e_t(G_{t}^{-1}(s))ds,  \label{meth4}
\end{equation}%
\begin{equation}
K(t)=H_{\pi }^{-1}(t)K_{c}(t)-H_{c}(t)H_{\pi }^{-2}(t)K_{\pi }(t)
\label{meth3}
\end{equation}%
\begin{equation}
g_{c,t}(\cdot )=c(G_{t}(Z(t)),G_{t}(f_{t}(\cdot )))\gamma_t (f_{t}(\cdot )),%
\text{ }g_{\pi ,t}=\pi (G_{t}(Z(t)),G_{t}(f_{t}(\cdot )))e_t(f_{t}(\cdot
))+K(t)e_t(f_{t}(\cdot))  \label{meth2}
\end{equation}%
\noindent\ and 
\begin{equation*}
\nu _{c,t}(y)=\frac{\partial c}{\partial y}(G_{t}(Z(t)),G_{t}(f_{t}(y)))%
\gamma_t (f_{t}(y)),\nu _{\pi ,t}(y)=\frac{\partial \pi }{\partial y}%
(G_{t}(Z(t)),G_{t}(f_{t}(y)))e_t(f_{t}(y)).
\end{equation*}%
Define 
\begin{equation*}
g_{t}=H_{\pi }^{-1}(t)g_{c,t}-H_{c,t}(t)H_{\pi ,t}^{-2}g_{\pi ,t}
\end{equation*}%
and 
\begin{equation*}
\nu _{t}=H_{\pi }^{-1}(t)\nu _{c,t}-H_{c}(t)H_{\pi }^{-2}(t)\nu _{\pi ,t}.
\end{equation*}

\noindent Then we have, uniformly in $t\in \lbrack 0,T]$, the following
representation, as $n\rightarrow \infty ,$%
\begin{equation}
\sqrt{n}(J_{n}(t)-J(t))=\alpha _{t,n}(g_{t})+\beta _{n}(\nu
_{t},t)+o_{P}^{\ast }(1),  \label{repA}
\end{equation}
with 
\begin{equation*}
\alpha _{t,n}(g_{t})=\frac{1}{\sqrt{n}}\sum_{j=1}^{n}\left\{ g_{t}(Y_{j}(t))-%
\mathbb{E}g_{t}(Y_{j}(t))\right\}
\end{equation*}
and 
\begin{equation}
{\small \beta _{n}(\nu _{t},t)=\frac{1}{\sqrt{n}}\sum \left\{
G_{t,n}(Y_{j}(t))-G_{t}(Y_{j}(t))\right\} \nu _{t}(Y_{j}(t)).}
\label{proc01}
\end{equation}
Suppose that (HS3), (HR1)-(HR3) and (HL1)-(HL4) hold. Then (\ref{repA})
holds with 
\begin{equation}
K(t)=K_{c}(t),g_{t}=g_{c,t}\text{ and }v_{t}=\nu _{c,t}
\end{equation}
\end{theorem}

\bigskip

\noindent This theorem expresses our studied time-dependent statistics as the sum of a
functional empirical process and the stochastic process (\ref{proc01})$.$ It
will be seen, for a fixed time, that $\beta _{n}(\nu _{t},t)$ is
asymptotically an integral of the quantile process $\sqrt{n}(s-V_{t,n}(s))$
based on $G_{t}(Y_{1}(t)),...,G_{t}(Y_{n}(t))$ (where $V_{t,n}(s)$ is the
empirical quantile function) and then of empirical process $\alpha _{t,n}(s)=%
\sqrt{n}\sum (I_{(G_{t}(Y_{j}(t)\leq s)}-s).$ These facts make easy the
handling of $\sqrt{n}(J_{n}(t)-J(t))$ in the modern empirical process
setting as stated in \cite{vaart}. We still need a thorough study of (\ref%
{proc01}) and its connection with $\alpha _{t,n}$ while the computation of
the variance and covariance function. This is done separately in \cite{logs} to avoid lengthy papers.\\

\noindent Now, we use these tools to give first, general laws for the WMLG
statistic below and then for the Kakwani class of indices in Section \ref%
{sec2} and for the Shorrocks-Thon indices in Section \ref{sec3}. We finish
by a special study of the absolute and the relative poverty changes in
Section \ref{sec4}.\\

\bigskip

\noindent While we deal with the general index and we use the outcomes of Theorem \ref%
{theo1}, we adopt the following writing : 
\begin{equation*}
\alpha _{t,n}(g_{t})=\frac{1}{\sqrt{n}}\sum_{j=1}^{n}W_{j}(t)-\mathbb{E}%
W_{j}(t)
\end{equation*}
where $W_{j}(t)=g_{t}(Y_j(t))$. Then we are entitled to express the hypotheses 
$(HT1)$ and $(HT2)$ below on the $W_{j}(t)$ in place of the $Y_{j}(t)$ for
the general case. And we suppose that $W_{j}(t)$ admits a density
probability $m_{t}$ for each $t\in \lbrack 0,T]$. In particular cases, we
will turn back to hypotheses on the $Y_{j}(t)$ for establishing (HT2) and
(HT3) and subsequently recover the results. \textbf{In the sequel, $r$ is a fixed
positive real number such that $0<r<1/2$.} And from now, the limits and the $o_{P}(1)$
are performed when $n\rightarrow \infty.$

\begin{itemize}
\item[(HT1)] For 0$\leq s,t\leq T,$ for some constant $K$, $\mathbb{E}\left|
W(t)-W(s)\right| ^{2}\leq K\left| t-s\right| ^{1+r}.$

\item[(HT2)] For 0$\leq s,t\leq T,$ for some constant $K$, $\left| \mathbb{E}%
(W(t))-\mathbb{E}(W(s))\right| ^{2}\leq K\left| t-s\right| ^{1+r}.$
\end{itemize}

In order to define our last assumption, we need the following functions : 
\newline
\begin{equation*}
g(\nu _{t},\nu _{s},t,s)=\int \left( \int_{x\geq u}\nu
_{t}(x)\,dG_{t}(x)\right) \left( \int_{y\geq v}\nu _{s}(y)\,dG_{s}(y)\right)
dG_{t,s}(u,v)
\end{equation*}
and 
\begin{equation*}
c(\nu _{t},t)=\int \left( \int_{x\geq u}\nu _{t}(x)\,dG_{t}(x)\right)
^{2}dG_{t}(u)
\end{equation*}
with, by convention, $\mathbb{E}_{t}h=\int h(u)dG_{t}(u)$ for a function $h.$
Set

\begin{itemize}
\item[(HT3)] If there is a universal constant $K_{5}$, such that for any $%
\delta >0,$ for large enough values of $n,$%
\begin{equation*}
\left| s-t\right| \leq \delta \Longrightarrow \left| 2(c(\nu _{t},t)-g(\nu
_{t},\nu _{s},t,s))+\left\{ (\mathbb{E}_{t}G_{t}\nu _{t})(\mathbb{E}%
_{s}G_{s}\nu _{s})-(\mathbb{E}_{t}G_{t}\nu _{t})^{2}\right\} \right|
\end{equation*}

\begin{equation}
\leq \frac{3}{2}K_{3}\left| s-t\right| ^{1+r}.  \label{unif04}
\end{equation}
\end{itemize}

\bigskip

We are now able to give our general main result.

\begin{theorem}
\label{theo2} Assume the conditions of Theorem \ref{theo1} hold and that
(HT1)-(HT3) are satisfied. Then the stochastic process $\{\sqrt{n}%
(J_{n}(t)-J(t)),0\leq t\leq T\}$ converges in $\ell ^{\infty }([0,T])$ to a
centered Gaussian process $\{\mathbb{G}(t),0\leq t\leq T\}$ with covariance
function 
\begin{equation*}
\Gamma (t,s)=\Gamma _{1}(g_{t},g_{s},t,s)+\Gamma _{2}(\nu _{t},\nu
_{s},t,s)+\Gamma _{2}(t,s)+\Gamma _{3}(g_{t},\nu _{t},t,s),
\end{equation*}%
\noindent with 
\begin{equation*}
\Gamma _{1}(g_{t},g_{s},t,s)=\int (g_{t}(x)-\eta (t))(g_{s}(y)-\eta
(s))dG_{t,s}(x,y),
\end{equation*}%
\begin{equation*}
\Gamma _{2}(\nu _{t},\nu _{s},t,s)=g(\nu _{t},\nu _{s},t,s)-(\mathbb{E}%
_{t}G_{t}\nu _{t})(\mathbb{E}_{s}G_{s}\nu _{s}),
\end{equation*}%
\begin{equation*}
g(\nu _{t},\nu _{s},t,s)=\int \left( \int_{x\geq u}\nu
_{t}(x)dG_{t}(x)\right) \left( \int_{x\geq v}\nu _{s}(x)dG_{s}(x)\right)
dG_{t,s}(u,v)
\end{equation*}%
and 
\begin{equation*}
\Gamma _{3}(g_{t},\nu _{s},t,s)=\kappa (g_{t},\nu _{s},t,s)+\kappa
(g_{s},\nu _{t},s,t)
\end{equation*}%
with 
\begin{equation*}
\kappa (g_{t},\nu _{s},t,s)=\int g_{t}(u)\left( \int_{x\geq v}\nu
_{s}(x)dG_{s}(x)\right) dG_{t,s}(u,v),
\end{equation*}%
and $g_{t}$ and $\nu _{t}$ are given in Theorem \ref{theo1}, and 
\begin{equation*}
\eta (t)=\int g_{t}(y)\:dG_{t}(y).
\end{equation*}
\end{theorem}

\begin{proof}

We have to do three things. First, we show that $\sqrt{n}(J_{n}(t)-J(t))$ is
asymptotically tight. Next, we have to prove that it converges in finite distributions. And
finally, we should compute the covariance function. We will only sketch the first and
the second tasks with the appropriate citations. The second will be properly
adressed.\\

\noindent Since the assumptions of Theorem \ref{theo1} hold, we have the
representation (\ref{repA}). Put 
\begin{equation}
N_{n}(t)=\alpha _{t,n}(g_{t})+\beta _{n}(\nu _{t},t).  \label{f01}
\end{equation}
First $(HT1)$ and $(HT2)$ yield, for each $j\in \lbrack 1,n],$ for some
constant $K,$ 
\begin{equation*}
\mathbb{E}\left| W_{j}(t)-W_{j}(s)\right| ^{2}+\left| \mathbb{E}W_{j}(t)-%
\mathbb{E}W_{j}(s)\right| ^{2}\leq K\left| s-t\right| ^{1+r},
\end{equation*}
and hence, by repeated use of $c_{2}$-inequality (that is, for any couple $%
(a,b)$ of scalars $\left| a+b\right| ^{2}\leq 2(\left| a\right| ^{2}+\left|
b\right| ^{2})$, for some constant $K,$ 
\begin{equation}
\left| \alpha _{t,n}(g_{t})-\alpha _{s,n}(g_{s})\right|^2 \leq K\left|
s-t\right| ^{1+r}.  \label{f02}
\end{equation}

\noindent We remind again that $r$ is strictly less that $1/2$, otherwise functions satisfying \ref{f02} are constant. Here and in the sequel, $K$ is a generic constant eventually taking different values
from one formula to another.\ Next, we find in \cite{logs}, that $%
\mathbb{E}(\beta _{n}(\nu _{t},t)-\beta _{n}(\nu _{s},s))^{2}$ is 
\begin{equation}
2(c(\nu _{t},t)-g(\nu _{t},\nu _{s},t,s))+\left\{ (\mathbb{E}_{t}G_{t}\nu
_{t})(\mathbb{E}_{s}G_{s}\nu _{s})-(\mathbb{E}_{t}G_{t}\nu _{t})^{2}\right\}
+\frac{K_{n}(t,s)}{n},  \label{f03}
\end{equation}
where $K_{n}(t,s)$ is bounded uniformly in $n,t$ and $s$.\ So by combining (%
\ref{f01}), (\ref{f02}), (\ref{f03}) and $(HT3),$ and by the $c_{2}$%
-inequality, we get for some $K$ that for any $\delta >0,$ for large enough
values of $n,$%
\begin{equation}
\left| s-t\right| \leq \delta \Longrightarrow \left|
N_{n}(t)-N_{n}(s)\right|^2 \leq K_{4}\left| s-t\right| ^{1+r}.
\end{equation}
Thus $\{N_{n}(t),t\in \lbrack 0,T]\}$ is asymptotically tight by Lemma 1 in 
\cite{sall-lo}, which is an adaptation of Example 2.2.12 in \cite{vaart}. To
finish the proof, we have to establish that finite-distributions of $%
N_{n}(t) $ converge to those of some Gaussian tight process $\mathbb{G}$.
For simplicity's sake, we do it in the two dimensional case, for $%
(N_{n}(t_{1}),N_{n}(t_{2})).$ Consider $N=aN_{n}(t_{1})+bN_{n}(t_{2}).$
Still for simpicity's sake, let us set 
\begin{equation*}
N_{n}(t_{1})=\frac{1}{\sqrt{n}}\sum_{j=1}^{n}g_{1}(X_{j})+\frac{1}{\sqrt{n}}%
\sum_{j=1}^{n}(G_{1,n}(X_{j})-G_{1}(X_{j}))\nu _{1}(X_{j})
\end{equation*}
and 
\begin{equation*}
N_{n}(t_{2})=\frac{1}{\sqrt{n}}\sum_{j=1}^{n}g_{2}(Y_{j})+\frac{1}{\sqrt{n}}%
\sum_{j=1}^{n}(G_{2,n}(X_{j})-G_{2}(Y_{j}))\nu _{2}(Y_{j}),
\end{equation*}
where the $(X_{j},Y_{j})^{\prime }s$ stand for the $%
(Y_{j}(t_{1}),Y_{j}(t_{2}))^{\prime }s$ as independent observations of $%
(X,Y),G_{1,n}$ (resp. $G_{2,n})$ is the empirical function based on $%
X_{1},...,X_{n}$ (resp. $Y_{1},...,Y_{n}$). Put 
\begin{equation*}
G(x,y)=P(X\leq x,Y\leq y),G_{1}(x)=G(x,+\infty ),G_{2}(y)=G(+\infty ,y).
\end{equation*}
Now let, for each $n\geq 1,$ $\varepsilon _{1,n}$ (resp. $\varepsilon
_{2,n}) $ be the quantile processes based respectively on $%
G_{1}(X_{1}),G_{1}(X_{2}),...,G_{1}(X_{n})$ (resp. $%
G_{2}(Y_{1}),G_{2}(Y_{2}),...,G_{2}(Y_{n}))$. It is not hard to see that 
\begin{equation*}
n^{-1/2}\sum_{j=1}^{n}(G_{1,n}(X_{j})-G_{1}(X_{j}))\nu
_{1}(X_{j})=\int_{0}^{1}-\varepsilon _{1,n}(s)\nu
_{1}(G_{1}^{-1}(s))ds+o_{P}(1)
\end{equation*}
and 
\begin{equation*}
n^{-1/2}\sum_{j=1}^{n}(G_{2,n}(Y_{j})-G_{2}(Y_{j}))\nu
_{2}(Y_{j})=\int_{0}^{1}-\varepsilon _{2,n}(s)\nu
_{2}(G_{2}^{-1}(s))ds+o_{P}(1).
\end{equation*}
Now let $\alpha _{1,n}$ and $\alpha _{2,n}$ be the empirical processes based
respectively on $G_{1}(X_{1}),$ $G_{1}(X_{2}),$ $...,$ $G_{1}(X_{n})$ and on $%
G_{2}(Y_{1}),G_{2}(Y_{2}),...,G_{2}(Y_{n}).$ We have (see \cite{shwell},
p.584) that $\alpha _{i,n}(s)=-\varepsilon _{i,n}(s)+o_{P}(1)$ uniformly in $%
s\in (0,1),$ which gives 
\begin{equation*}
\varepsilon _{n}(s,t)=(\varepsilon _{1,n}(s),\varepsilon _{2,n}(t))=-(\alpha
_{1,n}(s),\alpha _{2,n}(t))+o_{P}(1),
\end{equation*}
uniformly in $(s,t)\in (0,1)^{2}$. Now let us consider the functional
empirical process $\alpha _{n}$\ based on the $%
N_{j}=(G_{1}(X_{j}),G_{2}(Y_{j}))^{\prime }s,$ that is 
\begin{equation*}
\alpha _{n}(f)=\frac{1}{\sqrt{n}}\sum_{j=1}^{n}f(N_{j})-\mathbb{E}f(N_{j}),
\end{equation*}
where $f$ a real function defined on $(0,1)^{2}$ such that $\mathbb{E}%
f(N_{i})^{2}<\infty .$ Finally, let $\gamma _{n}$ the fonctional empirical
process based on the $(X_{i},Y_{i})^{\prime }s$ , defined for $h:\mathbb{R}%
^{2}\rightarrow \mathbb{R},$%
\begin{equation*}
\gamma _{n}(h)=\frac{1}{\sqrt{n}}\sum_{j=1}^{n}\left\{ h(X_{j},Y_{j})-%
\mathbb{E}h(X_{j},Y_{j})\right\} .
\end{equation*}

\noindent We have 
\begin{equation*}
\alpha _{n}(f)=\gamma _{n}(h_{f}),
\end{equation*}
for $h_{f}(x,y)=f(G_{1}(x),G_{2}(y)).$ We have by the classical results of
empirical process that $\{\gamma _{n}(h),$ $h\in \mathcal{H}\}$ converges to
a Gaussian process $\{G(h),h\in \mathcal{H}\}$ whenever $\mathcal{H}$ is a
Donsker class. It follows that $\{\gamma _{n}(1_{C}),$ $C\in \mathcal{C}\}$
converges to a Gaussian process $\{G(1_{C}),C\in \mathcal{C}\}$ whenever $%
\mathcal{C}$ is a Vapnik-Cervonenkis class. But $\mathcal{C}=\{1_{]-\infty
,x]\times ]-\infty ,y]},(x,y)\in \mathbb{R}^{2}\}$ is $VC$-class of index
not greater than 2. (see \cite{vaart} for $VC$-classes use to empirical
processes). Thus putting $h_{x,y}=1_{]-\infty ,x]\times ]-\infty ,y]},$ we
have 
\begin{equation*}
\gamma _{n}(x,y)=\gamma _{n}(h_{x,y})\leadsto G(h_{x,y})=\mathbb{G}(x,y),
\end{equation*}
in $\ell ^{\infty }(\mathbb{R}^{2})$ where $\mathbb{G}$ is a tight Gaussian
process such that 
\begin{equation*}
\mathbb{EG}(h_{1})\mathbb{G(}h_{2}))=\int (h_{1}(x,y)-\mathbb{E}%
_{h_{1}})(h_{2}(x,y)-\mathbb{E}_{h_{2}})dG(x,y).
\end{equation*}

\noindent Further, for $f_{1,s}=1_{[0,s]\times \lbrack 0,1]},$ $%
h_{f_{1,s}}(x,y)=1_{[0,s]\times \lbrack
0,1]}(G_{1}(x),G_{2}(y))=1_{]-\infty ,G_1^{-1}(s)]\times \mathbb{R}%
}(x,y)=h_{G^{-1}_1(s),+\infty }(x,y)$,

\begin{equation*}
\alpha _{1,n}(s)=\alpha _{n}(f_{1,s})=\gamma _{n}(h_{G^{-1}_{1}(s),+\infty })
\end{equation*}

\noindent and for $h_{f_{2,t}}(x,y)=h_{+\infty ,G^{-1}_2(t)}(x,y)$%
\begin{equation*}
\alpha _{2,n}(t)=\gamma _{n}(h_{+\infty ,G^{-1}_2(t)}(x,y))
\end{equation*}
Now, by using the Skorohod-Wichura-Dudley Theorem, we are entitled to
suppose that we are on a probability space such that 
\begin{equation*}
\sup_{(x,y)\in R^{2}}\left| \gamma _{n}(h_{x,y})-\mathbb{G}(h_{x,y})\right|
\rightarrow _{P}0.
\end{equation*}
Now, since the functions $\nu _{i}$\ are bounded, and putting $%
h_{1}(x,y)=g_{1}(x)$ and $h_{2}(x,y)=g_{2}(y),$ $%
N=aN_{n}(t_{1})+nN_{n}(t_{2})$ is equal to 
\begin{equation*}
{\small \left\{ a\mathbb{G}(h_{1})+b\mathbb{G}(h_{2})\right\} }
\end{equation*}
\begin{equation*}
+\int_{0}^{1}\left\{ a\mathbb{G}(h_{G^{-1}_1(s),+\infty })\nu
_{1}(G_{1}^{-1}(s))+b\mathbb{G}(h_{+\infty ,G^{-1}_2(s)})\nu
_{2}(G_{2}^{-1}(s))\right\} ds+o_{p}(1).
\end{equation*}
One easily proves that 
\begin{equation*}
\left\{ a\mathbb{G}(h_{1})+b\mathbb{G}(h_{2})\right\}
\end{equation*}
\begin{equation*}
+\int_{0}^{1}\left\{ a\mathbb{G}(h_{G^{-1}_1(s),+\infty })\nu
_{1}(G_{1}^{-1}(s))+b\mathbb{G}(h_{+\infty ,G^{-1}_2(s)})\nu
_{2}(G_{2}^{-1}(s))\right\} ds
\end{equation*}
is a Gaussian random variable since the second term is a Riemann integral,
which is a limit of finite linear combinations of Gaussian random variables.
Thus $N$ is asymptotically Gaussian. We are able to do the same for an
arbitrary finite-distribution $(N_{n}(t_{1}),...,N_{n}(t_{k})).$ The
computation of the limiting Gaussian process requires heavy calculations
done in \cite{logs}. The proof ends with providing the covariance
function $\Gamma _{1}$ of $\alpha _{t,n}(g_{t})$, $\Gamma _{1}$ of \ $\beta
_{n}(\nu _{t},t)$ and the covariance $\Gamma _{3}$ function between them.
\end{proof}

\subsection{Special cases}

Since the results are stated in a more general form and may appear very sophisticated, it seems necessary to show how they work for
common cases. We apply our results to two key examples in Welfare analysis : the class of Kakwani's and Shorrocks' statistics. These two examples are particularly interesting since they put the emphasis on the less deprived individuals within the whole population (with weight $n-j+1$) for Shorrock's statistic), or  within the marked individuals (with weight $Q-j+1$) for Kakwani's class of statistic including sen's measure). In both case, taking the weight at the power $k$ may lead to more accuracy in the statistical estimation.

\subsection{The Kakwani case}

\label{sec2}

We are now applying the general results to the Kakwani WMLG statistics of
parameter $k\geq 1$, defined by 
\begin{equation*}
J_{k,n}(t)=\frac{Q_{n}(t)}{n\sum_{j=1}^{Q_{n}(t)}j^{k}}%
\sum_{j=1}^{Q_{n}(t)}(Q_{n}(t)-j+1)^{k}d(\frac{Z(t)-Y_{j,n}(t)}{Z(t)}).
\end{equation*}
The way we are using here is to be repeated for any particular index. For
instance, the results in \cite{sall-lo} and \cite{lodurban} may be
rediscovered in this way. In this specific case, we turn the hypotheses
(HT1) and (HT2) on $W_{j}^{\prime }s$ to the $Y_{j}^{\prime }s$ as follows.
Suppose the $d.f.$ $G_{t}(x)=\mathbb{P}(Y_{j}(t)\leq x)$ admits a derivative 
$m_{t}(x).$ Put $G_{s,t}(u,v)=\mathbb{P}(X(t)\leq u,X(s)\leq v).$ Introduce :

\begin{itemize}
\item[(H0)] For $0\leq s,t\leq T,$ for some constant $K$, $%
\left|Z(s)-Z(t)\right|^2 \leq K\left| t-s\right| ^{1+r}.$

\item[(H1)] There exists a positive function $m$\ such that for $0\leq
s,t\leq T$, $u\in \mathbb{R},$ $0 < r < 1)$
\begin{equation*}
\left\vert m_{t}(u)-m_{s}(u)\right\vert \leq m(u)\left\vert t-s\right\vert
^{(1+r)/2}\text{ and }\int_{0}^{Z_{2}}m(u)du=K_{1}<\infty .
\end{equation*}%
and 
\begin{equation*}
\sup_{x\in (Z_{1},Z_{2})}\left\vert m(x)\right\vert =M_{0}<\infty .
\end{equation*}

\item[(H2)] For $0\leq s,t\leq T,$ for some constant $K$, 
\begin{equation*}
\sup_{u\geq 0}\left| G_{t,s}(u,u)-G_{s}(u)\right| \leq K\left| t-s\right|
^{1+r}
\end{equation*}
and 
\begin{equation*}
\left| G_{t}(Z(t))-G_{t}(Z(t)\wedge Z(s))\right| \leq K\left| t-s\right|
^{1+r}
\end{equation*}

\item[(H3)] For 0$\leq s,t\leq T,$ for some constant $K$, $\mathbb{E}\left|
Y(t)-Y(s)\right| ^{2}\leq K\left| t-s\right| ^{1+r}.$
\end{itemize}

\bigskip

We check, in the Kakwani case, that the representation of Theorem \ref{theo1}%
\ holds with $h(Q_{n}(t),n)=n^{k}$, $c(x,y)=(x-y)^{k}$, $\pi (x,y)=\frac{%
y^{k}}{x}$ and then 
\begin{equation*}
H_{\pi }(t)=G_{t}(Z(t))^{k}/(k+1),
\end{equation*}%
\begin{equation*}
H_{c}(t)=\int_{0}^{Z(t)}(G_{t}(Z(t))-G_{t}(y))^{k}\gamma _{t}({y)}dG_{t}(y),
\end{equation*}%
so that 
\begin{equation*}
J_{k}(t)=H_{c}(t)/H_{\pi
}(t)=(k+1)\int_{0}^{Z(t)}(1-G_{t}(y)/G_{t}(Z(t)))^{k-1}\gamma _{t}({y)}%
dG_{t}(y).
\end{equation*}%
Next 
\begin{equation*}
K_{c}(t)=k\int_{0}^{Z(t)}(G_{t}(Z(t))-G_{t}(y))^{k-1}\gamma
_{t}(y)dG_{t}(y),K_{\pi }(t)=-G_{t}(Z(t))^{k-1}/(k+1),
\end{equation*}%
and then 
\begin{equation*}
K(t)=(k+1)k\left\{ G_{t}(Z(t))^{-k-1}r_{k-1}(t)+r_{k}(t))\right\} ,
\end{equation*}%
where 
\begin{equation*}
r_{k}(t)=\int_{0}^{Z(t)}(G_{t}(Z(t))-G_{t}(y))^{k}\gamma _{t}(y)dG_{t}(y).
\end{equation*}%
For 
\begin{equation*}
g_{t}(\cdot )=(k+1)(1-G_{t}(f_{t}(\cdot ))/G_{t}(Z(t)))^{k}\gamma_t
(f_{t}(\cdot ))+K(t)1_{(f_{t}(\cdot )\leq Z(t))}.
\end{equation*}%
and 
\begin{equation*}
\nu _{t}(y)=-k(k+1)G_{t}(Z(t))^{-2k-1}
\end{equation*}%

\begin{equation*}
\times \left\{
(G_{t}(Z(t))^{k+1}(G_{t}(Z(t))-G_{t}(y))^{k-1}\frac{Z(t)-y}{Z(t)}%
-(k+1)r_{k}(t)G_{t}(y)^{k-1}\right\} 1_{(f_{t}(\cdot )\leq Z(t))}.
\end{equation*}%

\noindent we will get the representation 
\begin{equation*}
\sqrt{n}(J_{k,n}(t)-J_{k}(t))=N_{n}(t)+o_{P}(1)
\end{equation*}%
with 
\begin{equation*}
N_{n}(t)=\alpha _{t,n}(g_{t})+\beta _{n}(\nu _{t},t)
\end{equation*}

\begin{theorem}
\label{theo3} Let (HL1), (HL3), (HL4) , (H0)-(H3) hold. Then $\{\sqrt{n}%
(J_{n}(t)-J(t)),0\leq t\leq T\}$ converges in $\ell ^{\infty }([0,T])$ to a
centered Gaussian process with covariance function $\Gamma $ given in
Theorem \ref{theo2}
\end{theorem}

\bigskip

\begin{proof}
We begin to remark that $(H3)$ ensures that $\{\sqrt{n}%
(G_{t,n}(Z(t))-G_{t}(Z(t))),0\leq t\leq T\}$ is asymptocially tight and
hence $(HL2)$. It is then enough to show that $(HT1)$ and $(HT2)$ hold from $%
(H0),$ $(H1),(H2)$ and $(H3)$. But this follows from routine calculations
that we only sketch here. We place these calculation in the appendix.
\end{proof}

\section{The Shorrocks-Thon-like case}

\label{sec3}

We apply our results to the Shorrocks-Thon WMLG statistics measures defined by 
\begin{equation*}
J_{n}(t)=\frac{1}{n^{2}}\sum_{j=1}^{Q_{n}(t)}(2n-2-j+1)\text{ }d(\frac{%
Z(t)-Y_{j,n}(t)}{Z(t)}).
\end{equation*}
This is the Thon index. One obtains the Shorrocks one by replacing $1/n^{2} $
by $1/(n(n+1).$ We also check here that representation of Theorem \ref{theo1}
holds in the simple case corresponding to $(HS3)$, with $c(x,y)=(x-y).$ In
this case, $\pi $ is useless. Then

\begin{equation*}
J(t)=H_{c}(t)=2\int_{0}^{Z(t)}(1-G_{t}(y))\gamma _{t}(y)dG_{t}(y),
\end{equation*}
\begin{equation*}
K(t)=K_{c}(t)=0,\text{ \ }\nu (y)=\nu _{c}(y)=-2\gamma _{t}(y).
\end{equation*}
Here again $\{\sqrt{n}(J_{n}(t)-J(t)),0\leq t\leq T\}$ has the same
asymptotic behaviour described in Theorem \ref{theo3} with 
\begin{equation*}
g_{t}(y)=2(1-G_{t}(y))\text{ }and\text{ }\nu (y)=\nu _{c}(y)=-2\gamma _{t}(y)
\end{equation*}

\noindent under the same hypotheses (HL1)-(HL4) and (H0)-(H3)

\section{Estimation of the WLMG statistic variation}

\label{sec4}

Although they are very expensive to collect, longitudinal data are highly preferred
for adequate estimate of the absolute index variation $\Delta
J(t,s)=J(s)-J(t),$ which is the exact measure of WMLG change between
the periods $t$\ and $s$ and the associate relative WMLG variation $\Delta
RJ(t,s)=(J(s)-J(t))/J(t)$. Their respective natural estimators are of course  $\Delta J_{n}(t,s)=J_{n}(s)-J_{n}(t)$ and $\Delta
RJ_{n}(t,s)=(J_{n}(s)-J_{n}(t))/J_{n}(t).$ Our previous results yield the
follow

\begin{theorem}
Under the assumptioms of Theorem 1 or Theorem 2, 
\begin{equation*}
\sqrt{n}(\Delta J_{n}(t,s)-\Delta J(t,s))\rightarrow \mathcal{N}(0,\Gamma
_{4}(s,t)),
\end{equation*}%
where $\Gamma _{4}(s,t)=\Gamma (t,t)+\Gamma (s,s)-2\Gamma (t,s),$ and 
\begin{equation*}
\sqrt{n}(\Delta RJ_{n}(t,s)-\Delta RJ(t,s))\rightarrow \mathcal{N}(0,\Gamma
_{5}(t,s)).
\end{equation*}%
where 
\begin{equation*}
\Gamma _{5}(t,s)=a_{1}^{2}\Gamma (t,t)+a_{2}^{2}\Gamma
(s,s)+2a_{1}a_{2}\Gamma (s,t)
\end{equation*}%
with 
\begin{equation*}
a_{1}=-(1+\Delta RJ(t,s))/J(t)
\end{equation*}%
\begin{equation*}
a_{2}=1/J(t).
\end{equation*}
\end{theorem}

The proof is straightforward. We also might consider the convergence of $%
\sqrt{n}(\Delta J_{n}(t,s)-\Delta J(t,s))$ to the Gaussian process $\Delta 
\mathbb{G}(t,s)=\mathbb{G}(s)-\mathbb{G}(t)$ in $\ell ^{\infty }([0,T]^{2}).$
Anyway for fixed $t$ and $s,$ $\sqrt{n}(\Delta J_{n}(t,s)-\Delta J(t,s))$
converges to the Gaussian random variable $\Delta \mathbb{G}(t,s)=\mathbb{G}%
(t)-\mathbb{G}(s)$ by the continuity Theorem with $\Gamma _{4}(s,t)$ as
variance. Also, by using the Skorohod-Wichura-Dudley Theorem, we have 
\begin{equation*}
\sqrt{n}((\Delta RJ_{n}(t,s)-\Delta RJ(t,s))=a_{2}\mathbb{G}(s)+a_{2}%
\mathbb{G}(s)+o_{p}(1)
\end{equation*}%
An important application of the second part of this theorem is related to checking 
the achievement of specific goals. One may, within a national or regional strategy, whish to have
some deprivation limited to some extent. For example, the UN has assigned a number of goals,
named Millennium Development Goals (MDG), to its members. We are concerned here by one of  them. 
Indeed, it is whished to halve the extreme poverty in the world in year $s=2015$ starting from year $t=2000$.
When the WMLG statistic is a poverty measure, we may use $\Delta RJ(t,s))$ and check whether it is less than $-0.5$. And an $(1-\alpha )-$confidence
interval $IR(\alpha )$ based on these results is 
\begin{equation*}
\lbrack \Delta RJ_{n}(t,s)-n^{-1/2}\sqrt{\Gamma _{5}(s,t)}u_{1-\alpha
/2},\Delta RJ_{n}(t,s)+n^{-1/2}\sqrt{\Gamma _{5}(s,t)}u_{1-\alpha /2}]
\end{equation*}%
\begin{equation*}
\equiv \lbrack J^{0}(\alpha ),J^{1}(\alpha )],
\end{equation*}%
where $\mathbb{P}(\mathcal{N}(0,1)\leq u_{1-\alpha /2})=\alpha .$ This MDG
will be reported achieved at the 95$\%$ level if the number $J^{1}(\alpha
)\leq -0,5.$

\subsection{Datadriven applications and variance computations}

We apply our results in Economics and Welfare analysis. Especially, we consider the
household surveys in Senegal in 2011 (ESAM II) and in 2006 (EPS) from which we construct pseudo-panel data
and apply our results. 

\subsubsection{Variance computations for Senegalese data}

We apply our results to Senegalese data. We do not really have longitudinal
data. So we have constructed pseudo-panel data of size $n = 116$, from two
surveys: ESAM II conducted from 2001 to 2002 and EPS from 2005 to 2006. We
get two series $X^{1}$ and $X^{2}$. We present below the values of $%
\Gamma_{I}(1,2)$ denoted here $\gamma(1)$, $\Gamma_{J}(1,2)$ denoted here $%
\gamma(2)$ and $\Gamma(1,2)$ denoted here $\gamma(3)$.\newline

\noindent When constructing pseudo-panel data, we get small sizes like $n=116$
here. We use these sizes to compute the asymptotic variances in our results
with nonparametric methods. In real contexts, we should use high sizes
comparable to those of the real databases, that is around ten thousands,
like in the Senegalese case. Nevertheless, we back on medium sizes, for
instance $n=696,$ which give very accurate confidence intervals as shown in
the tables below.\\

\noindent Before we present the outcomes, let us say some words on the
packages. We provide different R script files at:\\

\begin{center}
\textit{http://www.ufrsat.org/lerstad/resources/sallmergslo01.zip}
\end{center}

\bigskip

\noindent The user should already have his data in two files \textit{%
data1.txt} and \textit{data2.txt}. The first script file named after \textit{%
gamma$_{-}$mergslo1.dat} provides the values of $\gamma(1)$, $\gamma(2)$ and 
$\gamma(3)$ for the FGT measure for $\alpha=0,1,2$ and for the six
inequality measures used here. The second script file named \textit{gamma$%
_{-}$mergslo2.dat} performs the same for the Shorrocks measure. Finally, 
\textit{gamma$_{-}$mergslo3.dat} concerns the kakwani measures. Unless the
user uploads new \textit{data1.txt} and \textit{data2.txt} files, the
outcomes should the same as those presented in the Appendix.\newline

\subsubsection{Analysis}

First of all, we find that, at an asymptotical level, all our inequality
measures and poverty indices used here have decreased. 
%except the Atkinson one.
When inspecting the asymptotic variance, we see that for the poverty index,
the FGT and the Kakwani classes respectively for $\alpha=1$, $\alpha=2$ and $%
k=1$ and $k=2$ have the minimum variance, specially for $\alpha=2$ and $k=2$.
This advocates for the use of the Kakwani and the FGT measures for poverty
reduction evaluation.

\begin{table}[h!]
\begin{tabular}{l||c|c|c}
\hline
Index J & $\Delta J(1,2)$ & $\Gamma_4(1,2)$ & $CI_{95\%}(\Delta J(1,2))$ \\ 
\hline\hline
SHOR & $-0.03024621$ & $0.02353406$ & $[-0.04264967,-0.01985518]$ \\ \hline
KAK$(1)$ & $-0.02108905$ & $0.01097123$ & $[-0.02982085, -0.01425729]$ \\ 
\hline
KAK$(2)$ & $-0.02055594$ & $0.01007820$ & $[-0.02961271,-0.01469601]$ \\ 
\hline
FGT$(0)$ & $-0.05977098$ & $0.3170756$ & $[-0.09355847, -0.009889805]$ \\ 
\hline
FGT$(1)$ & $-0.01859332$ & $0.00922992$ & $[ -0.02620413, -0.01192899]$ \\ 
\hline
FGT$(2)$ & $-0.00432289$ & $0.0008381113$ & $[-0.007194404, -0.002892781]$
\\ \hline\hline
\end{tabular}%

\bigskip

\caption{Variations of the poverty indices}
\label{var-pov-ind}
\end{table}

\section{Conclusion}

We obtained asymptotic laws of the UTB WMLG statistics with in mind, among
other targets, the uniform estimation of the variation $\Delta J(t,s)$ and the relative variation $\Delta RJ(t,s)$. The results are only
illustrated  with simple datadriven applications to income databases in Senegal. This opens large datadriven application in whole economic areas.
In the theoritical hand, the Lower Threshold Based weighted mean loss or gain statistics is to be studied in accordance with heavy tail conditions
and to be applied in Insurance and HIV/VIH fields.

%biblio

\newpage

\section{Appendix}

Put 
\begin{equation*}
W(t)=g_{t}(Y(t))=W_{1}(t)+W_{2}(t)
\end{equation*}
with 
\begin{equation*}
W_{1}(t)=K(t)1_{(Y(t)\leq Z(t))}
\end{equation*}
and 
\begin{equation*}
W_{2}(t)=(k+1)(1-G_{t}(Y(t))/G_{t}(Z(t)))^{k}\gamma _{t}(Y(t)).
\end{equation*}
We have first to prove that for $i=1,2$, 
\begin{equation*}
\mathbb{E}\left| W_{i}(s)-W_{i}(t)\right| ^{2}\leq K\left| t-s\right| ^{1+r}
\end{equation*}
Based on the expression of $K(t)$ and on the facts that $r_{k}(t)$ and $%
G_{t}(u)$ for $Z_{1}\leq u\leq Z_{2}$ are uniformly bounded for $t\in
\lbrack 0,T],$ it suffices to prove that 
\begin{equation}
\left| G_{t}(Z(t))^{-k}-G_{s}(Z(s))^{-k}\right| \leq K\left| s-t\right|
^{(1+r)/2}  \label{f001}
\end{equation}
for $k\geq 1,$%
\begin{equation}
\left| r_{k}(t)-r_{k}(s)\right| \leq K\left| s-t\right| ^{(1+r)/2}
\label{f002}
\end{equation}
and 
\begin{equation}
\mathbb{E}\left| 1_{(Y(t)\leq Z(t))}-1_{(Y(s)\leq Z(s)}\right| ^{2}\leq
K\left| t-s\right| ^{1+r}.  \label{f003}
\end{equation}
This would help to conclude with the $c_{2}-inequality$ that$_{\text{ }}$%
\begin{equation}
\mathbb{E}\left| W_{1}(s)-W_{1}(t)\right| ^{2}\leq K\left| t-s\right| ^{1+r}.
\label{F1}
\end{equation}
Let us establish (\ref{f01}). We have 
\begin{equation*}
\left| G_{t}(Z(t))^{-k}-G_{s}(Z(s))^{-k}\right| \leq \left|
G_{t}(Z(t))^{-k}-G_{t}(Z(s))^{-k}\right| +\left|
G_{t}(Z(s))^{-k}-G_{s}(Z(s))^{-k}\right|
\end{equation*}
\begin{equation*}
\leq k\left| Z(t)-Z(s)\right| m_{t}(Z(s,t)G_{t}(Z(s,t))^{-k-1}+k\left|
G_{t}(Z(s))-G_{s}(Z(s))\right| B(s,t)^{-k-1},
\end{equation*}
where $Z(s,t)$ lies between $Z(t)$ and $Z(s)$ and $B(s,t)$ lies between $%
G_{t}(Z(s))$ and $G_{s}(Z(s)).$ We then get 
\begin{equation*}
\left| G_{t}(Z(t))^{-k}-G_{s}(Z(s))^{-k}\right| \leq k\beta ^{-2k}\zeta
^{k-1}K\left\{ \left| Z(t)-Z(s)\right| +\left|
G_{t}(Z(s))-G_{s}(Z(s))\right| \right\}
\end{equation*}
\begin{equation}
\leq K\left| s-t\right| ^{1+r/2}.  \label{app01}
\end{equation}
Now we show (\ref{f02}) 
\begin{equation*}
\left| r_{0}(t)-r_{0}(s)\right| =\left| \int_{0}^{Z(t)}\gamma
_{t}(u)m_{t}(u)du-\int_{0}^{Z(s)}\gamma _{s}(u)m_{s}(u)du\right|
\end{equation*}
\begin{equation*}
\leq \int_{0}^{Z(s)\wedge Z(t)}\left| \gamma _{t}(u)m_{t}(u)-\gamma
_{s}(u)m_{s}(u)\right| +\int_{Z(s)\wedge Z(t)}^{Z(s)}\gamma _{t}(u)m_{t}(u)du
\end{equation*}
\begin{equation*}
+\int_{Z(s)\wedge Z(t)}^{Z(s)}\gamma _{s}(u)m_{s}(u)du
\end{equation*}
Since $\gamma _{s}$ is uniformly bounded, we have by $(H0)$ and $(H1),$%
\begin{equation}
\left| \int_{Z(s)\wedge Z(t)}^{Z(t)}\gamma
_{t}(u)m_{t}(u)du+\int_{Z(t)\wedge Z(s)}^{Z(s)}\gamma
_{s}(u)m_{s}(u)du\right| \leq K\left| s-t\right| ^{1+r/2}  \label{app02}
\end{equation}
Further 
\begin{equation}
\int_{0}^{Z(s)\wedge Z(t)}\left| \gamma _{t}(u)m_{t}(u)-\gamma
_{s}(u)m_{s}(u)\right| \leq \int_{0}^{Z_{2}}\left| \gamma _{t}(u)-\gamma
_{s}(u)\right| m_{t}(u)du  \label{app03}
\end{equation}
\begin{equation*}
+\int_{0}^{Z_{2}}\gamma _{s}(u)\left| m_{t}(u)-m_{s}(u)\right| du,
\end{equation*}
and, since $\gamma _{t}(x)=d((Z(t)-u)/Z(t)),$ we get that 
\begin{equation}
\int_{0}^{Z(s)\wedge Z(t)}\left| \gamma _{t}(u)-\gamma _{s}(u)\right|
m_{t}(u)du\leq  \label{app04}
\end{equation}
\begin{equation*}
\int_{0}^{Z(s)\wedge Z(t)}\left| Z(t)-Z(s)\right| \text{ }B(s,t)^{-1}\text{ }%
u\text{ }d^{\prime }((Z(s,t)-u))/Z(s,t))\text{ }m_{t}(u)du,
\end{equation*}
where $Z(s,t)$ lies between $Z(t)$ and $Z(s).$ Then 
\begin{equation}
\int_{0}^{Z(s)\wedge Z(t)}\left| \gamma _{t}(u)-\gamma _{s}(u)\right|
m_{t}(u)du\leq K_{1}\beta ^{-1}Z_{2}^{2}\text{ }\left| Z(t)-Z(s)\right| \leq
K\left| s-t\right| ^{(1+r)/2}.  \label{app05}
\end{equation}
From (\ref{app03})-(\ref{app05}), we conclude that 
\begin{equation*}
\left| r_{0}(t)-r_{0}(s)\right| \leq K\left| s-t\right| ^{(1+r)/2}.
\end{equation*}
and for $k>1,$%
\begin{equation*}
r_{k}(t)=\int_{0}^{Z(t)\wedge Z(s)}(G_{t}(y)-G_{t}(Z))^{k}\gamma
_{t}(u)m_{t}(u)du
\end{equation*}
\begin{equation*}
+\int_{Z(t)\wedge Z(s)}^{Z(t)}(G_{t}(y)-G_{t}(Z))^{k}\gamma _{t}(u)m_{t}(u)du
\end{equation*}
and 
\begin{equation*}
r_{k}(s)=\int_{0}^{Z(t)\wedge Z(s)}(G_{s}(y)-G_{s}(Z))^{k}\gamma
_{s}(u)m_{s}(u)du
\end{equation*}
\begin{equation*}
+\int_{Z(t)\wedge Z(s)}^{Z(s)}(G_{s}(y)-G_{s}(Z))^{k}\gamma _{s}(u)m_{s}(u)du
\end{equation*}
with 
\begin{equation*}
\left| \int_{Z(t)\wedge Z(s)}^{Z(t)}(G_{t}(y)-G_{t}(Z))^{k}\gamma
_{t}(u)m_{t}(u)du\right|
\end{equation*}
and 
\begin{equation*}
\left| \int_{Z(t)\wedge Z(s)}^{Z(s)}(G_{s}(y)-G_{s}(Z))^{k}\gamma
_{s}(u)m_{s}(u)du\right|
\end{equation*}
less than $2M_{0}\beta ^{k}\left| Z(t)-Z(s)\right| .$ Now $\left|
r_{k}(t)-r_{k}(s)\right| $ is less than $A+B$ with 
\begin{equation*}
A=\int_{0}^{Z(t)\wedge Z(s)}\left|
(G_{t}(y)-G_{t}(Z))^{k}-(G_{s}(y)-G_{s}(Z))^{k}\right| \text{ }\gamma
_{s}(u)m_{s}(u)du
\end{equation*}
and 
\begin{equation*}
B=\int_{0}^{Z(t)\wedge Z(s)}(G_{t}(y)-G_{t}(Z))^{k}\text{ }\left| \gamma
_{s}(u)m_{s}(u)-\gamma _{t}(u)m_{t}(u)\right| du.
\end{equation*}
By (H2), $A$ is less than $2kZ_{2}\xi ^{k-1}(\int_{0}^{Z_{2}}m(u)du)\left|
s-t\right| ^{(1+r)/2}$ and $B\leq K\left| s-t\right| ^{(1+r)/2}$ by (\ref%
{app03}), (\ref{app04}) and (\ref{app05}). Then for $k\geq 1,$%
\begin{equation*}
\left| r_{k}(t)-r_{k}(s)\right| \leq K\left| s-t\right| ^{(1+r)/2},
\end{equation*}
which proves (\ref{f02}). Let us finally prove (\ref{f03}). We have by $(H2),$
for a fixed z, 
\begin{equation*}
\mathbb{E}\left| 1_{(Y(t)\leq z))}-1_{(Y(s)\leq z)}\right| ^{2}\leq \left|
G_{t}(z)-G_{t,s}(z,z))\right| +\left| G_{s}(z))-G_{t,s}(z,z)\right| \leq
K\left| s-t\right| ^{1+r},
\end{equation*}
for some constant $K$. Then by the $c_{2}$-inequality, 
\begin{equation*}
\mathbb{E}\left| 1_{(Y(t)\leq Z(t))}-1_{(Y(s)\leq Z(s))}\right| ^{2}\leq 2%
\mathbb{E}\left| 1_{(Y(t)\leq Z(t))}-1_{(Y(s)\leq Z(t))}\right| ^{2}
\end{equation*}
\begin{equation*}
+2\mathbb{E}\left| 1_{(Y(s)\leq Z(t))}-1_{(Y(s)\leq Z(s))}\right| ^{2}
\end{equation*}
with 
\begin{equation*}
\mathbb{E}\left| 1_{(Y(t)\leq Z(t))}-1_{(Y(s)\leq Z(t))}\right| ^{2}\leq
\left| G_{t}(Z(t)))-G_{t,s}(Z(t),Z(t)))\right|
\end{equation*}
\begin{equation*}
+\left| G_{s}(Z(t)))-G_{s}(Z(s)\wedge Z(t)\right| \leq K\left| s-t\right|
^{1+r}
\end{equation*}
and 
\begin{equation*}
\mathbb{E}\left| 1_{(Y(s)\leq Z(t))}-1_{(Y(s)\leq Z(s))}\right|
^{2}=G_{t}(Z(s)+G_{s}(Z(t))-2G_{s}(Z(t)\wedge Z(s))\leq K\left| s-t\right|
^{1+r}
\end{equation*}
and then (\ref{f03}) holds.

By putting together (\ref{f01}), (\ref{f02}) and (\ref{f03}) and by
repeatedly using the $c_{2}$-inequality, we arrive at (\ref{F1}).

Now we have to establish that 
\begin{equation}
\mathbb{E}\left| W_{2}(s)-W_{2}(t)\right| ^{2}\leq K\left| t-s\right| ^{1+r}.
\label{F02}
\end{equation}
Put 
\begin{equation*}
W_{2}(t)=(k+1)A(t)B(t)
\end{equation*}
with $A(t)=(1-G_{t}(Y(t))/G_{t}(Z(t))^{k}$, $%
B(t)=d((Z(t)-Y(t))/Z(t))1_{(Y(t)\leq Z(t))}$. We have by readily check that 
\begin{equation*}
\left| A(t)-A(s)\right| \leq 2\beta ^{-1}\left| Y(t)-Y(s)\right| +M_{0}\beta
^{-2}\left| Z(t)-Z(s)\right| +\beta ^{-2}\left|
G_{s}(Z(t))-G_{t}(Z(t))\right| .
\end{equation*}
Then by (H0)-(H3) and the $c_{2}-$inequality 
\begin{equation*}
\mathbb{E}\left| A(t)-A(s)\right| ^{2}\leq K\left| t-s\right| ^{1+r}.
\end{equation*}
\ Next 
\begin{equation*}
B(t)=d((Z(t)-Y(t))/Z(t))^{k}(1_{(Y(t)\leq Z(t)\wedge Z(s))}+1_{(Z(t)\wedge
Z(s)\leq Y(t)\leq Z(t))})
\end{equation*}
\begin{equation*}
=B_{1}(t,s)+B_{2}(t,s)
\end{equation*}
and 
\begin{equation*}
B(s)=B_{1}(s,t)+B_{2}(s,t)
\end{equation*}
with 
\begin{equation*}
B_{1}(t,s)=d((Z(t)-Y(t))/Z(t))^{k}1_{(Y(t)\leq Z(t)\wedge Z(s))}
\end{equation*}
and 
\begin{equation*}
B_{2}(t,s)=d((Z(t)-Y(t))/Z(t))^{k}1_{(Z(t)\wedge Z(s)\leq Y(t)\leq Z(t))}.
\end{equation*}
Then by (H2) 
\begin{equation*}
\mathbb{E}B_{2}(t,s)=G_{t}(t)-G_{t}(Z(t)\wedge Z(s))\leq K\left| t-s\right|
^{1+r},
\end{equation*}
and 
\begin{equation*}
\mathbb{E}B_{2}(s,t)\leq K\left| t-s\right| ^{1+r}.
\end{equation*}
Next, by putting $C(s,t)=1_{(Z(t)\wedge Z(s)\leq Y(t)\leq Z(t))},$%
\begin{equation*}
\left| B_{1}(t,s)-B_{1}(s,t)\right| =C(s,t)d^{\prime }(D(s,t))\left| \frac{%
Z(t)-Y(t)}{Z(t)}-\frac{Z(s)-Y(s)}{Z(s)}\right| ,
\end{equation*}
where $D(s,t)$ lies between ($Z(t)-Y(t))/Z(t)$ and ($Z(s)-Y(s))/Z(s).$ We
finally get 
\begin{equation*}
\left| B_{1}(t,s)-B_{1}(s,t)\right| \leq (Z_{1}^{-1}+Z_{2}Z_{1}^{-2})\left|
Z(t)-Z(s)\right| +Z_{1}^{-1}\left| Y(t)-Y(s)\right| .
\end{equation*}
By similar methods, we get 
\begin{equation*}
\mathbb{E}\left| B_{1}(t,s)-B_{1}(s,t)\right| ^{2}\leq K\left| t-s\right|
^{1+r}.
\end{equation*}
By combining all that precedes, we get (\ref{F02}), which together with (\ref{F1})
establishes by the $c_{2}-$inequality 
\begin{equation}
\mathbb{E}\left| W(s)-W(t)\right| ^{2}\leq K\left| t-s\right| ^{1+r}.
\label{F00}
\end{equation}

\bigskip

Now we have to prove that 
\begin{equation*}
\left| \mathbb{E}W(s)-\mathbb{E}W(t)\right| ^{2}\leq K\left| t-s\right|
^{1+r}.
\end{equation*}

We only sketch this second part. Let us consider $W_{i}(t)$, $i=1,2.$ We
have 
\begin{equation*}
\mathbb{E}W_{1}(t)=K(t)\int_{0}^{Z(t)}m_{t}(u)\text{ }du
\end{equation*}
and 
\begin{equation*}
\mathbb{E}W_{2}(t)=(k+1)\int_{0}^{Z(t)}(1-G_{t}(u)/G_{t}(Z(t)))^{k}\text{ }d(%
\frac{Z(t)-u}{Z(t)})\text{ }m_{t}(u)\text{ }du.
\end{equation*}
By (\ref{f01}),(\ref{f02}) and the decomposition of $\left| \gamma
_{t}(t)-\gamma _{s}(s)\right| $ used in (\ref{app04}), we have 
\begin{equation*}
\left| K(t)-K(s)\right| \leq K\left| s-t\right| ^{1+r/2}.
\end{equation*}
Furthermore 
\begin{equation*}
\int_{0}^{Z(t)}m_{t}(u)du-\int_{0}^{Z(s)}m_{s}(u)du=\int_{Z(t)\wedge
Z(s)}^{Z(t)}m_{t}(u)du-\int_{Z(t)\wedge Z(s)}^{Z(s)}m_{s}(u)du
\end{equation*}
\begin{equation*}
+\int_{0}^{Z(t)\wedge Z(s)}m_{t}(u)-m_{s}(u)du.
\end{equation*}
We then get 
\begin{equation*}
\left| \int_{0}^{Z(t)}m_{t}(u)du-\int_{0}^{Z(s)}m_{s}(u)du\right|
\end{equation*}
\begin{equation*}
\leq 2M_{0}\left| Z(t)-Z(s)\right| +Z_{2}K\left| s-t\right| ^{1+r/2}.
\end{equation*}
Then 
\begin{equation*}
\left| \mathbb{E}W_{1}(s)-\mathbb{E}W_{1}(t)\right| ^{2}\leq K\left|
s-t\right| ^{1+r}.
\end{equation*}
Now 
\begin{equation*}
\mathbb{E}W_{2}(t)=\int_{0}^{Z(t)}S(t,u)\text{ }du
\end{equation*}
with 
\begin{equation*}
S(t,u)=(k+1)(1-G_{t}(u)/G_{t}(Z(t)))^{k}\text{ }d(\frac{Z(t)-u}{Z(t)})\text{ 
}m_{t}(u).
\end{equation*}
Then 
$$
\mathbb{E}W_{2}(t)-\mathbb{E}W_{2}(t)=
$$ 
\begin{equation*}
\int_{0}^{Z(t)}S(t,u)du-\int_{0}^{Z(s)}S(s,u)du=\int_{Z(t)\wedge
Z(s)}^{Z(t)}S(t,u)-\int_{Z(t)\wedge Z(s)}^{Z(s)}S(s,u)du
\end{equation*}
\begin{equation*}
+ \int_{0}^{Z(t)\wedge Z(s)}S(t,u)-S(s,u)\text{ }du.
\end{equation*}
Since $S(t,u)$ is uniformly bounded, we have 
\begin{equation*}
0\leq \left| \int_{Z(t)\wedge Z(s)}^{Z(t)}S(t,u)\text{ }du-\int_{Z(t)\wedge
Z(s)}^{Z(s)}S(s,u)\text{ }du\right| \leq K\left| s-t\right| ^{1+r/2}.
\end{equation*}
Moreover, one easily shows by the (H0)-(H3), with similar techniques used
when handling $r_{k}(t)$, that 
\begin{equation*}
\left| S(t,u)-S(s,u)\right| \leq K\left| s-t\right| ^{1+r/2}.
\end{equation*}
Thus 
\begin{equation*}
\left| \mathbb{E}W_{2}(s)-\mathbb{E}W_{2}(t)\right| ^{2}\leq K\left|
s-t\right| ^{1+r}.
\end{equation*}

\end{document}